\documentclass[journal]{IEEEtran}
\makeatletter
\def\ps@headings{%
\def\@oddhead{\mbox{}\scriptsize\rightmark \hfil \thepage}%
\def\@evenhead{\scriptsize\thepage \hfil \leftmark\mbox{}}%
\def\@oddfoot{}%
\def\@evenfoot{}}
\makeatother

\usepackage{amssymb,amsmath}
\usepackage{paralist}
\usepackage{graphicx}
\usepackage{color}
\usepackage{cite}

\newcommand{\F}{\mathbf{F}}

\newcommand{\C}{\mathcal{C}}
\newcommand{\N}{\mathcal{N}}

\newtheorem{theorem}{\textbf{Theorem}}
\newtheorem{lemma}[theorem]{\textbf{Lemma}}
\newtheorem{definition}[theorem]{\textbf{Definition}}

\newcommand{\nix}[1]{}

\begin{document}
\title{Encoding of Network Protection Codes Against  Link and Node Failures Over Finite Fields}
\author{
\authorblockN{Salah A. Aly~~~~~ and~~~~ Ahmed E. Kamal\\}
\authorblockA{Department of Electrical and Computer Engineering\\ Iowa State University, Ames, IA 50011, USA\\ Email: \{salah,kamal\}@iastate.edu}
 }
 
 \maketitle

\begin{abstract}
Link and node failures  are common two fundamental problems that affect operational networks.
Hence, protection of communication networks  is essential to increase their reliability, performance, and operations. Much research work has been done to protect against link and node failures,  and to provide reliable solutions based on pre-defined provision or dynamic restoration of the domain. In this paper we  develop network protection strategies  against  multiple link   failures using network coding and joint capacities.    In these strategies, the source nodes apply network coding for their transmitted data to provide backup copies for recovery at the receivers' nodes. Such techniques can be applied to optical, IP, and mesh networks. The encoding operations of protection codes are defined over finite fields. Furthermore, the normalized capacity of the communication network is given by $(n-t)/n$ in case of  $t$ link failures. In addition, a bound on the minimum required field size is derived.

\end{abstract}

\section{Introduction}\label{sec:intro}
With the increase in the capacity of backbone networks, the failure of
a single link or node can result in the loss of enormous amounts of
information, which may lead to catastrophes, or at least loss of
revenue. Network connections are therefore provisioned with the property that they can survive such failures, and hence several techniques  have been introduced in the literature. Such
techniques either add extra resources, or reserve some of the
available network resources as backup circuits, just for the sake of
recovery from failures.
Recovery from failures is also required to be agile in order to
minimize the network outage time.
This recovery usually involves two steps: fault diagnosis and
location, and rerouting connections.
Hence, the optimal network survivability problem is
a multi-objective problem in terms of resource efficiency, operation cost, and agility~\cite{zeng07}.

In network survivability, the four different types of failures that might affect network operations are\cite{somani06,zhou00}: \begin{inparaenum} \item  link failure,
\item node failure, \item shared risk link group (SRLG) failure, and
    \item network control system failure.
  \end{inparaenum}
Henceforth, one needs to design network protection strategies against
these types of failures.  Although the common frequent failures are link
failures, node failures sometimes happen due to burned swritch/router,
fire, or any other hardware damage. In addition, the failure might be due to
network maintenance.

 Network coding allows the intermediate nodes not only to forward packets using network
scheduling algorithms, but also encode/decode them using algebraic
primitive operations, see~\cite{ahlswede00,fragouli06,soljanin07,yeung06}
and the references therein. As an application of network coding, data loss because of failures in communication links can be detected and recovered if the sources are allowed to perform network coding operations.

Recently,  network protection strategies against multiple link failures using
network coding and reduced capacities are proposed in~\cite{aly08i, kamal07a}. In this paper, we provide a new technique for protecting network failures
using \emph{protection codes} and \emph{reduced capacity} in which the encoding operations are defined over finite fields. This technique
can be deployed at an overlay layer in optical mesh networks, in which
detecting failure is an essential task. The benefits of this approach are
that:
\begin{compactenum}[i)]
\item
It allows receivers to recover the lost data without contacting a
third parity or main domain server.
\item
It has less computational complexity and does not require adding
extra paths.
\item
All $n$ disjoint paths have full capacity except at $t$ paths in case of
protecting against $t$ link  failures.
\end{compactenum}

\goodbreak
This paper is organized as follows.  In Sections~\ref{sec:networkmodel} and \ref{sec:terminolgoy} we present the
network model and problem definition. In Section~\ref{sec:Tfilures} we provide network protections against $t$ link  failures. We present differentiated distributed capacities in Section~\ref{sec:distributedcapacities}, and demonstrate analysis of protection codes  in Section~\ref{sec:analysis}. Finally, Bounds on the finite field size is proved in Section~\ref{sec:boundfeilds}, and the paper is concluded in Section~\ref{sec:conclusion}.

\section{Network Model and Assumptions}\label{sec:networkmodel}

In this section we introduce the network model and provide the needed
assumptions.  The main hypothesis of this network model can be stated as
follows.

\begin{figure}[t]
 \begin{center} 
  \includegraphics[height=4.5cm,width=7.2cm]{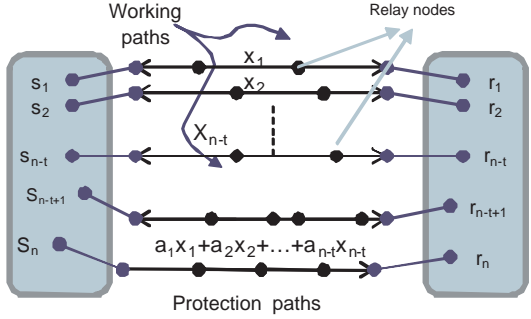}
  \caption{Network protection against a single path failure using reduced capacity and network coding. One path out of  $n$ primary paths  carries encoded data. The black points represent various other relay nodes}\label{fig:npaths}
\end{center}
\end{figure}

\begin{compactenum}[i)]
\item Let $\N$ be a network represented by an abstract graph
    $G=(\textbf{V},E)$, where $\textbf{V}$ is the set of nodes and $E$
    be   set of undirected edges. Let $S$ and $R$ are   sets of independent sources
    and destinations, respectively. The set $\textbf{V}=V\cup S \cup
    R$ contains the relay nodes, sources, and destinations. Assume for simplicity that $|S|=|R|=n$,
    hence the set of sources is equal to the set of receivers.

\item The node can be a router, switch, or an end terminal depending on
    the network model $\N$ and the transmission layer.

\item $L$ is a set of links $L=\{L_1,L_2,\ldots,L_n\}$ carrying the data
    from the sources to the receivers as shown in Fig.~\ref{fig:npaths}.
    All connections have the same bandwidth, otherwise a connection with
    high bandwidth can be divided into multiple connections, each of which
    has a unit capacity. There are exactly $n$ connections. For
    simplicity, we assume that the number of sources is less than or equal
    to the number of links. A sender with a high capacity can divide its
    capacity into multiple unit capacity, each of which has its own link. Put differently, \begin{eqnarray}
    \{(s_i,w_{1i}),(w_{1i},w_{2i}),\ldots,(w_{(\lambda)i},r_i) \},\end{eqnarray} where
    $1\leq i\leq n$ and $(w_{(j-1)i},w_{ji}) \in E$, for some integer $\lambda \geq 1$.
Hence
we have $|S|=|R|=|L|=n$. The n connection paths are pairwise link
disjoint.

\item The data from all sources are sent in cycles. Each cycle has a
    number of time slots $n$. Hence $t_j^\delta$ is a value at round time
    slot $j$ in cycle $\delta$.

\item The failure on a link $L_i$ may happen due to the network
    circumstance such as a link replacement, overhead, etc. We
    assume that the receiver is able to detect a failure and our
    protection strategy is able to recover it.

\item In this model $\N$, consider only a single link failure, it is
    sufficient to apply the encoding and decoding operation over a finite
    field with two elements, we denote it $\F_2=\{0,1\}$.
\end{compactenum}

\section{Problem Setup and Terminology}\label{sec:terminolgoy}

We assume that there is a set of n connections that need to be protected
with $\%100$ guaranteed against  single and multiple link failures. We assume that all
connections have the same bandwidth, and each link (one hop or circuit) has the same bandwidth as a path.

Every sender $s_i$ prepares a packet \emph{$packet_{s_i \rightarrow r_i}$}
 to send to a receiver $r_i$. The packet contains the sender's ID, data
$x_i^\ell$, and a round time for every cycle $t^\ell_\delta$ for some
integers $\delta$ and $\ell$. There are  two types of packets:
\begin{compactenum}[i)]
\item {\bf Plain Packets:} Packets sent without coding, in which the sender does not need
    to perform any coding operations. For example, in case of packets
    sent without coding, the sender $s_i$ sends the following packet
    to the receiver $r_i$.
\begin{eqnarray}
packet_{s_i \rightarrow r_i}:=(ID_{s_i},x_i^\ell,t_\delta^\ell)
\end{eqnarray}

\item  {\bf Encoding Packets:} Packets sent with encoded data, in which the sender $s_j$
    sends other sender's data. In this case, the sender $s_j$ sends the following packet to receiver
$r_j$:
\begin{eqnarray}
packet_{s_j \rightarrow r_j}:=(ID_{s_j},\sum_{i=1,j\neq i}^n \alpha_i x_i^\ell,t^\ell_\delta),
\end{eqnarray}
where $\alpha_i \in \F_q$.
\end{compactenum} In either case the sender has a full capacity in the
connection link $L_i$.

\begin{definition}\label{def:capacitylink}
The capacity of a connecting link $L_i$ between $s_i$ and $r_i$ is
defined by \begin{eqnarray}
c_i=\left\{
      \begin{array}{ll}
        1, & \hbox{$L_i$ has \emph{active signals};} \\
        0, & \hbox{otherwise.}
      \end{array}
    \right.
\end{eqnarray}
And the total capacity is given by the summation of all link
capacities. What we mean by an \emph{active}  link is that the receiver is able to receiver un-encoded signals/messages throughout this link and process them.
\end{definition}

Clearly, if all links are active then the total capacity is $n$ and
normalized capacity is $1$. In general the normalized capacity of the network
for the active and failed links is computed by
\begin{eqnarray}
C_\N=\frac{1}{n}\sum_{i=1}^n c_i.
\end{eqnarray}

The following definition describes the \emph{working} and
\emph{protection} paths between two network switches as shown in
Fig.~\ref{fig:npaths}.

\begin{definition}
The \emph{working paths} on a network with $n$ connection paths carry
un-encoded traffic under normal operations. The \emph{Protection paths} provide an
alternate backup path to carry encoded traffic. A
protection scheme ensures that data sent from the sources will reach the
receivers in case of failure incidences on the working paths.
\end{definition}


\section{NPS-T: Protecting Against $t$ Path Failures}\label{sec:Tfilures}
In this section we present a network protection strategy against $t$
failures in optical networks. Assume the same notations as shown in the
previous sections hold.  Assume also that the total number of failures are $t$ and they
happen at arbitrary $t$ links.

\begin{figure*}[t]
\begin{center}
\begin{eqnarray}\label{eq:tFailuresScheme}
\begin{array}{|c|cccccc|}
\hline
&1&2&\ldots&j&\ldots&m=\lceil n/t\rceil\\
\hline    \hline
s_1 \rightarrow r_1 & y_1&x_1^1 &\ldots&x_1^{j-1}&\ldots &x_{1}^{m-1} \\
s_2 \rightarrow r_2 &  y_2&  x_2^1 &\ldots&x_2^{j-1}&\ldots&x_2^{m-1}  \\
\vdots&\vdots&\vdots&\vdots&\vdots&\vdots&\vdots\\
s_{t} \rightarrow r_{t} &  y_{t}&  x_{t}^1 &\ldots&x_{t}^{j-1}&\ldots&x_{t}^{m-1}  \\
\!\! s_{t+1} \! \rightarrow \!\! r_{t+1} \!\! &\!\! x_{t+1}^1\!&\! y_{t+1}&\ldots&x_{2t+1}^3&\! \ldots&\! x_{2t+1}^{m-1} \!\! \\
\vdots&\vdots&\vdots&\vdots&\vdots&\vdots&\vdots\\
s_{2t} \rightarrow r_{2t}&x_{2t}^1&y_{2t}&\ldots& x_{2t}^3&\ldots&x_{2t}^{m-1}  \\
\vdots\ddots&\vdots\ddots&\vdots\ddots&\vdots\ddots&\vdots\ddots&\vdots\ddots&\vdots\ddots\\
\!\! s_{jt+\ell}\!\! \rightarrow \!\!r_{jt+\ell}\!\!&\!\!x_{jt+\ell}^1\!&\!x_{jt+\ell}^2&\ldots&\! y_{jt+\ell}^3& \ldots&\! x_{jt+\ell}^{m-1} \!\! \\
\vdots\ddots&\vdots\ddots&\vdots\ddots&\vdots\ddots&\vdots\ddots&\vdots\ddots&\vdots\ddots\\
\!\! s_{t(m-1)+1} \rightarrow r_{t(m-1)+1}&x_{t(m-1)+1}^1&x_{t(m-1)+1}^2& \ldots&x_{t(m-1)+1}^j&\! \ldots&\! y_{t(m-1)+1}  \\
\vdots&\vdots&\vdots&\vdots&\vdots&\vdots&\vdots\\
s_{mt} \rightarrow r_{mt} & x_{mt}^1&x_{mt}^2&\ldots&x_{mt}^j&\ldots&y_{mt}\\
\vdots\ddots&\vdots\ddots&\vdots\ddots&\vdots\ddots&\vdots\ddots&\vdots\ddots&\vdots\ddots\\
\hline
\end{array}
\end{eqnarray}
\caption{The encoding Scheme of $t$ link failures. $m=\lceil n/t\rceil$, $1 \leq j \leq m$ and $1 \leq \ell \leq t$. $t$ out of the $n$ connections carry encoded data. The coefficients are chosen over $\F_q$, for $q \geq n-t+1$.}
\end{center}
\end{figure*}

Let $m=\lceil n/t\rceil$, hence we have $m$ rounds per cycle. The encoding operations of NPS-T against $t$ failures are shown in Scheme~(\ref{eq:tFailuresScheme}).
We can see that $y_\ell$ in general  is given by

\begin{eqnarray}\label{eq:y_NPS-T}
y_\ell=\sum_{i=1}^{(j-1)t} a_i^\ell x_i^{j-1} + \sum_{i=jt+1}^n a_i^\ell x_i^j  \nonumber \\ \mbox{   for  } (j-1)t+1 \leq \ell \leq jt,  1 \leq j \leq n.
\end{eqnarray}

The advantages of NPS-T approach is that
\begin{itemize}
\item

The data is encoded and decoded online, and it will be sent and
received in different rounds. Once the receivers detect failures, they
are able to obtain a copy of the lost data immediately without delay
by querying the neighboring nodes with unbroken working paths.

\item The recovery is assured with $\%100$. Since $t$ paths will
    carry encoded data, up to $t$ failures can be recovered.

\item Using this strategy, no extra paths are needed. This will make
    this approach more suitable for applications, in which adding
    extra paths is not allowed.

\item Since in  real case scenarios, the number of failures is     very small in comparison to the number of working paths, the NPS-T performs well.

\item The encoding  operations are linear, and the
    coefficients of the variables $x_i^j$ are taken from a finite
    field with $q\geq n-t+1$ elements.
\end{itemize}

\begin{theorem}
Let $n$ be the total number of connections from sources to receivers. The
capacity of NPS-T strategy shown in Scheme~\ref{eq:tFailuresScheme} against $t$ path failures is given by
\begin{eqnarray}
\C_{\N}=(n-t)/(n)
\end{eqnarray}
\end{theorem}

\begin{lemma}
The encoding  Scheme~(\ref{eq:tFailuresScheme}) is optimal in terms of max capacity.
\end{lemma}
One can not find a better encoding scheme against $t$ link failures  rather than providing   one protection path against one failure. Indeed $t$ protection paths are used to protect $t$ link failures and this is shown in Scheme~(\ref{eq:tFailuresScheme}).


\subsection{Encoding Operations}
Assume that each connection path
$L_i$ (L) has a unit capacity from a source  $s_i$ (S) to a
receiver $r_i$ (R). The data sent from the sources S to the receivers R
is transmitted in rounds.
Under NPS-T, in every round $n-t$ paths are used to carry new data
($x_i^j$), and $t$ paths are used to carry protected data units.
there are $t$ protection paths.
Therefore, to treat all connections fairly,
there will be $ n/t$ rounds in a cycle, and in each round the
capacity is given by n-t.

We consider the case in which all symbols $x_i^j$  belong to the same round.
The first t sources transmit the first encoded data units
$y_1,y_2,\ldots,y_{t}$, and
in the second round, the next $t$ sources
transmit $y_{t+1},y_{t+2},\ldots,y_{2t}$, and so
on. All sources $S$ and receivers $R$ must keep
track of the round numbers. Let $ID_{s_i}$ and $x_{s_i}$ be the ID and
data initiated by the source $s_i$. Assume the round time $j$ in cycle
$\delta$ is given by $t^{j}_{\delta}$. Then the source $s_i$ will send
$packet_{s_i}$ on the working path which includes
\begin{eqnarray}
Packet_{s_i}=(ID_{s_i}, x_{i}^\ell, t^\ell_\delta)
\end{eqnarray}
Also, the source $s_j$, that transmits on a protection path, will
send a packet $packet_{s_j}$:
\begin{eqnarray}
Packet_{s_j}=(ID_{s_j}, y_j, t^\ell_\delta),
\end{eqnarray}
where $y_k$ is defined in~(\ref{eq:y_NPS-T}). Hence the protection
paths are used to protect the data transmitted in round $\ell$, which
are included in the $x^l_i$ data units.
 So, we have a system of $t$ independent equations at each round time
 that will be  used to recover at most $t$ unknown variables.

The strategy NPS-T is a generalization of protecting against a single path failure shown in the previous section in which $t$
protection paths are used instead of one protection path in case of
one failure. We also notice that most of the network operations suffer
from one and two path failures~\cite{zhou00,somani06}.

\subsection{Proper Coefficients Selection}
One way to select the coefficients $a_j^\ell$'s in each round such that we
have a system of $t$ linearly independent equations is by using the matrix H shown in~(\ref{bch:parity}). Let $q$ be the order of a finite field, and
$\alpha$ be the root of unity. Then we can use this matrix   to define the coefficients of the senders as:
\begin{eqnarray}\label{bch:parity} H =\left[
\begin{array}{ccccc}1&1&1&\ldots&1\\1 &\alpha &\alpha^2 &\cdots &\alpha^{n-1}\\1
&\alpha^2 &\alpha^4 &\cdots &\alpha^{2(n-1)}\\\vdots& \vdots &\vdots
&\ddots &\vdots\\1 &\alpha^{t-1} &\alpha^{2(t-1)} &\cdots
&\alpha^{(t-1)(n-1)}\end{array}\right].\end{eqnarray}
We have the following assumptions about the encoding operations.
\begin{compactenum}
\item  Clearly if we have one failure $t=1$, then all coefficients will be
    one. The first sender will always choose the unit value.

\item  If we assume $t$ failures, then the $y_1,y_2,\ldots,y_t$ equations
    are written as:
\begin{eqnarray}
y_1&=&\sum_{i=t+1}^nx_i^1, ~~~~~
y_2=\sum_{i=t+1}^n \alpha^{(i-1) }x_i^2,
\\
\label{eq:tcofficients}
y_j&=&\sum_{i=t+1}^n \alpha^{i(j-1) \mod (q-1)}x_i^\ell,
\end{eqnarray}
\end{compactenum}

This equation gives the general theme to choose the coefficients at any particular round in any cycle. However, the encoded data $y_i$'s are defined as shown in Equation~(\ref{eq:tcofficients}). In other words, for the first round in cycle one, the coefficients of the plain data $x_1,x_2,\ldots,x_t$ are set to zero.

\subsection{Decoding Operations} We know that the coefficients
$a_1^\ell,a_2^\ell,\ldots,a_n^\ell$ are elements of a finite field, hence
the inverses of these elements exist and they are unique.  Once
a node fails which causes $t$ data units to be lost, and once the
receivers receive $t$ linearly independent equations, they can
linearly solve these equations to obtain the unknown t data units.
At one
particular cycle j, we have three cases for the failures
\begin{compactenum}[i)]
\item All t link failures happened in the working paths, i.e. the working
    paths have failed to convey the messages $x_i^\ell$ in round $\ell$.
In this case, $n-t$ equations will be received, $t$ of which are linear
combinations of $n-t$ data units, and the remaining $n-2t$ are explicit
$x_i$ data units, for a total of $n-t$ equations in $n-t$ data units.
    In this case any t
    equations (packets) of the t encoded packets can be used to recover
    the lost data.

    \item All t link failures happened in the protection paths. In this case, the
exact remaining n-t packets are working paths and they do not experience
any failures. Therefore, no recovery operations are needed.
\item The third case is that the failure might happen in some working and
    protection paths simultaneously in one particular round in a cycle.
    The recover can be done using any t protection paths as shown in case
    i.
\end{compactenum}

\section{Bounds on the Finite Field Size, $\F_q$}\label{sec:boundfeilds}
In this section we derive lower and upper bound on the alphabet size
required  for the encoding and decoding operations.  In the proposed
schemes we assume that direction connections exist between the senders and
receivers, which the information can be exchanged with neglected cost.

The first result shows that the alphabet size required must be greater
than the number of connections that carry unencoded data.

\begin{theorem}
Let $n$ be the number of connections in the network model $\N$, then the
receivers are able to decode the encoded messages over $\F_q$ and will
recover from $t\geq 2$ path failures if
\begin{eqnarray}
q\geq n-t+1.
\end{eqnarray}
Also, if $q=p^r$, then $r \leq \lceil \log_p(n+1) \rceil$. The binary field is sufficient in case of a single path failure.
\end{theorem}
\begin{proof}
We will prove the lower bound by construction. Assume a NPS-T at one
particular time $t_\delta^\ell$ in the round  $\ell$ in a certain cycle
$\delta$.  The protection code of NPS-T against $t$ path failures is given in~\ref{bch:parity}.

Without loss of generality, the interpretation of
Scheme~(\ref{bch:parity}) is as follows:
\begin{compactenum}[i)]
\item The columns correspond to the senders $S$ and rows correspond to t
    encoded data $y_1,y_2,\ldots,y_t$.
\item The first row corresponds to $y_1$ if we assume the first
    round
    in
    cycle one. Furthermore, every row represents the coefficients of
    every senders at a particular round.
\item The column $i$ represents the coefficients of the sender $s_i$
    through all protection paths $L_1,L_2,\ldots,L_t$.
\item Any element $\alpha^i \in F_q$ appears once in a column and row,
    except in the follow column and first row, where all elements are one's.
    \item All columns (rows) are linearly independent.

\end{compactenum}

Due to the fact that the t failures might occur at any t working paths of
$L=\{l_1,L_2,\ldots,l_n\}$, then we can not predict the t protection paths
as well. This means that t out of the n columns do not participate in the
encoding coefficients, because t paths will carry encoded data. We notice
that removing any $t$ out of the $n$ columns in
Scheme~(\ref{bch:parity}) will result to $n-t$ linearly
independent columns. Therefore the smallest
    finite field that  satisfies this condition must have $n-t+1$
    elements.

  The upper bound comes from the case of no failures, hence $q \geq
    (n+1)$. Assume q is a prime power , then the result follows.
\end{proof}

if $q=2^r$, then in general the previous bound can be stated as
\begin{eqnarray}
n-t+1 \leq q \leq 2^{\lceil \log_2(n+1) \rceil}.
\end{eqnarray}
 The following result shows
the maximum admissible paths, which can suffer from failures, and the
decoding operations can be achieved successfully.

\begin{lemma}
Let $n$ and $t$ be the number of connections and failures in the network
model $\N$, then we have $t \leq \lfloor n/2  \rfloor$.
\end{lemma}
\begin{proof}
The proof is a direct consequence and from the fact that the protection
paths must be less than or equal to  the number of working paths.
\end{proof}
This lemma shows that one can not provide protection paths better than
duplicating the number of working paths.


\section{Network Protection Using Distributed Capacities and Network Coding}\label{sec:distributedcapacities}

In this section we develop network protection strategy where some connection paths have high priorities (less bandwidth, high demand). Let $n$ be  the set of available connections (disjoint paths from sources to receivers). Let $m$ be the set of rounds in every  cycle. In the previous strategies (NPS-T) we assumed that all connection paths have the same priority demand and working capacities. This might be the real case scenario.  connections that carry applications with multimedia traffic have high priority than applications that carry data traffic. Therefore, it is required to design network protection strategies based on the traffic and sender priorities.

Consider that available working connections $n$ may use their bandwidth
assignments in asymmetric ways. Some connections are less demanding in
terms of bandwidth requirements than other connections that require full capacity frequently. Therefore connections with less demanding can transmit more
protection packets, while other connections demand more bandwidth, and can therefore transmit fewer protection packets throughout transmission rounds. Let $m$ be the number of rounds and $t_i^\delta$ be the time of transmission in a cycle $\delta$ at round $i$. For a particular cycle $i$, let $t$ be the number of protection paths against $t$ failures that might affect the working paths. We will design network protection strategy against $t$ arbitrary link failures (NPS-T2) as follows. Let the source $s_j$ sends $d_i$ data packets and $p_i$ protection packets such that $d_j+p_j=m$. Put differently:

\begin{eqnarray}
\sum_{i=1}^n (d_i+p_i)=nm
\end{eqnarray}
In general we do not assume that $d_i =d_j$ and $p_i=p_j$. NPS-T2 is described as shown in Scheme~\ref{eq:tfailures2}.

\begin{figure}
\begin{eqnarray}\label{eq:tfailures2}
\begin{array}{|c|ccccccc|}
\hline
& \multicolumn{7}{|c|}{\mbox{ round time cycle 1 }} \\ \hline
\hline
&1&2&3&4&\ldots&m-1&m\\
\hline    \hline
s_1 \rightarrow r_1 & y_1^1&x_1^1 &x_1^2&y_1^2&\ldots &y_1^{p_1}&x_{1}^{d_1} \\
s_2 \rightarrow r_2 &  x_2^1& y_2^1 &x_2^2&x_2^3&\ldots&x_2^{d_2} &y_2^{p_2} \\
\vdots&\vdots&\vdots&\vdots&\vdots&\vdots&\vdots&\vdots\\
s_{i} \rightarrow r_{i} &  y_{i}^1&  x_{i}^1 &x_{i}^2&y_{i}^2&\ldots&y_i^{p_i}  &x_{i}^{d_i}\\
\vdots&\vdots&\vdots&\vdots&\vdots&\vdots&\vdots&\vdots\\
s_{j} \rightarrow r_{j} &  x_{j}^1&  x_{j}^2 &y_{j}^1&x_{j}^3&\ldots&x_{j}^{d_j}&y_j^{p_j}  \\
\vdots&\vdots&\vdots&\vdots&\vdots&\vdots&\vdots&\vdots\\
s_n \rightarrow r_n & x_n^1&y_n^1&x_n^2&x_n^4&\ldots&y_{n}^{p_n}&x_n^{d_n}\\
\hline
\end{array}
\end{eqnarray}
\end{figure}
The encoded data $y_i^\ell$ is given by
\begin{eqnarray}
y_i^\ell =\sum_{k=1,y_k^\ell \neq y_j^\ell}^n x_k^\ell
\end{eqnarray}

We assume that the maximum number of failures that might occur in a particular cycle is $t$. Hence the number of protection paths (paths that carry encoded data) is $t$. The selection of the working and protection paths in every round is done using a priority  demanding function at the senders's side. It will also depend on the traffic type and service provided on these protection and working connections.

In Scheme~(\ref{eq:tfailures2}) every connection $i$ is used to carry $d_i$ unencoded data $x_i^1,x_i^2,\ldots,x_i^{d_i}$ (working paths) and $p_i$ encoded data $y_i^1,y_i^2,\ldots,y_i^{p_i}$ (protection paths) such that $d_i+p_i=m$.

\begin{lemma}\label{lem:nps-t2}
Let $t$ be the number of connection paths carrying encoded data in every round in NPS-T2, then the normalized network capacity $C_\N$ is given by
\begin{eqnarray}
(n-t)/n
\end{eqnarray}
\end{lemma}
\begin{proof}
The proof is straight forward from the fact that $t$ protection paths exist in every round, hence $n-t$ working paths are available throughout all $m$ rounds.
\end{proof}

\section{Analysis of the Protection Codes Over $\F_q$}\label{sec:analysis}
We will prove correctness of the protection codes over $\F_q$.
Let $\F_q$ be a finite field with $q$ elements such that $q=p^r$ for some nonzero integer $r$  and prime $p$. We will drive a scheme to recover from any $m$ failures in the $n+m$ primary and protection paths. Assume $t$ be the number of failures in the primary paths. We have three cases

\begin{compactenum}[i)]
\item
All failures occur in the primary paths, $t=m$. In this case we need to establish a system of $t$ linearly independent equations in $t$ variables.

\item
$t$ failures occur in the primary paths and $m-t$ failures occur in the protection paths. In this case we need to establish a system of equations to recover the failures in the primary paths only.
\item All failures occur in the protection paths. No recovery process is needed in this case.

\end{compactenum}

We will show the encoding operation in case of directional connections from the senders to receivers.
consider the worst case scenario in which $m=t$. We can describe the encoding scheme for multiple link failures as shown in~(\ref{bch:parity}).

All $\alpha$'s powers are taken module the field size,  i.e. $\alpha^{ij \mod q=n+1}$. In other words, if $q\geq n+1$, then we have the encoding matrix

\begin{eqnarray}\label{eq:encodingscheme2}
\left[
\begin{array}{ccccc}1 &1 &1 &\cdots &1\\  \alpha &\alpha^2 &\alpha^3&\cdots &\alpha^{n}\\
\alpha^2 &\alpha^4 &\alpha^6 &\cdots &\alpha^{2(n)}\\
\vdots& \vdots &\vdots
&\ddots &\vdots\\ \alpha^{t-1} &\alpha^{2(t-1)}&\alpha^{3(t-1)} &\cdots
&\alpha^{(t-1)(n)}\end{array}\right]\end{eqnarray}

In this case we have $\alpha^{q-1}=1$, $q$ is a prime power.

The first column represents the coefficients of the encoding data at the first sender. Also, the first row represents the binary coefficients of all senders in case of a single link failure. Hence $\alpha^{i-1}$ column represents the coefficients of the encoding data at the $i$ sender for all $1\leq i \leq n-1$. 

In general for multiple $m=t$ failures, the encoding data in the j-th protection is given by
\begin{eqnarray}
y_{n+j}=\sum_{i=1}^n \alpha^{j(i-1) \mod q}  x_i,
\end{eqnarray}
for $1 \leq j \leq m$.

As a matter of fact, the square sub-matrix of $t$ columns of the encoding scheme $\ref{eq:encodingscheme2}$ is invertable (has a full rank) if and only if its determinant is not equal to zero~\cite{lint99}. We will show that for any $t$ arbitrary link failures, the receivers are able to form a system of $t$ linearly independent equations and recover the lost data.

\bigskip

\begin{lemma}\label{lemma:matrixfullrank}
If there are $t$ link failures in the primary paths, then the receivers are successfully able to recover from those failures  using $t$ protection paths.
\end{lemma}
\begin{proof}
Let \Big[$\begin{array}{ccccc}1&\alpha^{j_1}&\alpha^{2j_1}&\ldots&\alpha^{(t-1)j_1}\end{array}$  \Big] represent the any arbitrary column in the encoding scheme (\ref{eq:encodingscheme2}) indexed by the second element  $\alpha^{j_1}$. Choosing any $t$ arbitrary columns $\alpha^{j_1},\alpha^{j_2},\ldots,\alpha^{j_t}$ yield
\begin{eqnarray}\label{eq:matrix-t}
\left[\begin{array}{cccc}1&1&\ldots&1 \\ \alpha^{j_1}&\alpha^{j_2}&\ldots&\alpha^{j_t}\\
\alpha^{2j_1}&\alpha^{2j_2}&\ldots&\alpha^{2j_t} \\
\ldots&\ldots&\ldots&\ldots\\
\alpha^{(t-1)j_1}&\alpha^{(t-1)j_2}&\ldots&\alpha^{(t-1)j_t}\end{array}  \right]\end{eqnarray}
\end{proof}
Hence we have a system of $t$ equations in $t$ variables. Clearly, all elements in each row are different. Indeed this system has determinant given by the form~~\cite[Theorem 6.5.5]{lint99}
\begin{eqnarray}
\alpha^{j_1+j_2+j_3+\ldots+j_t}\prod _{h > \ell } \Big( \alpha^{j_h} -\alpha^{j_\ell}\Big) \neq 0,
\end{eqnarray}
which proves the result.

\bigskip

Now, we shall prove the general case that any $\mu \times \mu$ square sub-matrix of the matrix~(\ref{eq:encodingscheme2}) has a full rank.
Assume the square matrix is represented by
%
%

\begin{eqnarray}\label{eq:subencodingscheme}
B=\left[
\begin{array}{ccccc}
\alpha^{i_{1}j_{1}} &\alpha^{i_{1}j_{2}} &\cdots &\alpha^{i_{1}j_{\mu}}\\
\alpha^{i_2j_{1}}&\alpha^{i_2j_{2}} &\cdots &\alpha^{i_2j_{\mu}}\\
\vdots& \vdots
&\ddots &\vdots\\ \alpha^{i_{\mu}j_{ 1}} &\alpha^{i_{\mu}j_{2}} &\cdots
&\alpha^{i_{\mu}j_{\mu}} \end{array}\right]\end{eqnarray}

where $1\leq \mu \leq n$  and $\alpha^{i_{i'} j_{j'}} \in \F_q$.

\begin{lemma}
The sub-matrix $B$ described in (\ref{eq:subencodingscheme}) has a full rank.
\end{lemma}
\begin{proof}
We proceed the proof by mathematical induction.
\begin{compactenum}[i)]
\item We first prove that any $2 \times 2$ sub-matrix of $B$ has a full rank. It means that for any four elements lie in the corner are not alike (do not share a common factor). Put differently, $i \neq j$ and $\ell \neq 1$,
\begin{eqnarray}
\left[
\begin{array}{cc}
\alpha^{i} &\alpha^{j}\\
\alpha^{\ell i}&\alpha^{\ell j}
 \end{array}\right]\end{eqnarray}

If we divide the second row by $\alpha^{(1-\ell) i}$, we obtain $\alpha^{i}$. Now assume by contradiction that
$\alpha^{(1-\ell) i} .\alpha^{\ell j}=\alpha^{ j} $. Or
$\alpha^{(1-\ell) i}=\alpha^{(1-\ell) j} \mod q $. Obviously, this contradicts  the fact that $\ell \neq 1$ and $i \neq j$.  In addition $(l-1)(j-i)=0 \mod q$ contradicts the fact about the field order. Hence, the result is a consequence.

\item
Now, assume the matrix
\begin{eqnarray}\label{eq:subencodingscheme1}
B_{\mu-1}=\left[
\begin{array}{ccccc}
\alpha^{i_{1}j_{1}} &\alpha^{i_{1}j_{2}} &\cdots &\alpha^{i_{1}j_{\mu-1}}\\
\alpha^{i_2j_{1}}&\alpha^{i_2j_{2}} &\cdots &\alpha^{i_2j_{\mu-1}}\\
\vdots& \vdots
&\ddots &\vdots\\ \alpha^{i_{\mu-1}j_{ 1}} &\alpha^{i_{\mu-1}j_{2}} &\cdots
&\alpha^{i_{\mu-1}j_{\mu-1}} \end{array}\right]\end{eqnarray}
has a full rank.
\item
We will add any arbitrary row and column to the matrix  $B_{\mu-1}$ to construct the matrix $B$.

\begin{eqnarray}\label{eq:subencodingscheme1}
B=\left[
\begin{array}{cccc|c}
\alpha^{i_{1}j_{1}} &\alpha^{i_{1}j_{2}} &\cdots &\alpha^{i_{1}j_{\mu-1}}&\alpha^{i_{1}j_{\mu}}\\
\alpha^{i_2j_{1}}&\alpha^{i_2j_{2}} &\cdots &\alpha^{i_2j_{\mu-1}}&\alpha^{i_2j_{\mu}}\\
\vdots& \vdots
&\ddots &\vdots\\
\alpha^{i_{\mu-1}j_{ 1}} &\alpha^{i_{\mu-1}j_{2}} &\cdots
&\alpha^{i_{\mu-1}j_{\mu-1}}&\alpha^{i_{\mu-1}j_{\mu}}\\  \hline \alpha^{i_{\mu}j_{ 1}} &\alpha^{i_{\mu}j_{2}} &\cdots
&\alpha^{i_{\mu}j_{\mu-1}}&\alpha^{i_{\mu}j_{\mu}} \end{array}\right]\end{eqnarray}

All elements in the last columns are different, also all elements in the last row are different. Since $\alpha^{i_ij_{j}}$ is an element in $\F_q$, it has a unique inverse. Therefore, we can divide every row in the matrix B by the element in the last column. Hence, we have

\begin{eqnarray}\label{eq:subencodingscheme1}
B'=\left[
\begin{array}{cccc|c}
\alpha^{i_{1}'j_{1}'} &\alpha^{i_{1}'j_{2}'} &\cdots &\alpha^{i_{1}'j_{\mu-1}'}&1\\
\alpha^{i_2'j_{1}'}&\alpha^{i_2'j_{2}'} &\cdots &\alpha^{i_2'j_{\mu-1}'}&1\\
\vdots& \vdots
&\ddots &\vdots\\
\alpha^{i_{\mu-1}'j_{1}'} &\alpha^{i_{\mu-1}'j_{2}'} &\cdots
&\alpha^{i_{\mu-1}'j_{\mu-1}'}&1\\  \hline \alpha^{i_{\mu}'j_{ 1}'} &\alpha^{i_{\mu}'j_{2}'} &\cdots
&\alpha^{i_{\mu}'j_{\mu-1}'}&1 \end{array}\right]\end{eqnarray}
\end{compactenum}
All powers of $\alpha$'s are taken module $q$. Furthermore, all elements in each row (or column) are pairwise disjoint. The matrix $B'$ is similar to the matrix shown in~(\ref{eq:matrix-t}). Using lemma~\ref{lemma:matrixfullrank}, the matrix $B'$ has a full rank given by $\mu$.
\end{proof}

\bigskip

\section{Conclusion}\label{sec:conclusion}

In this paper we demonstrated the encoding operations of network protection codes defined over finite fields. We derived a bound on the minimum field size required for choosing unique coefficients of data sent on the working paths. In addition we presented a scheme for differentiated services in cases of some working paths have high priorities in terms of bandwidth and capacity assignments.

\scriptsize
\bibliographystyle{plain}

\bibliographystyle{ieeetr}

\end{document}